\documentclass[journal]{IEEEtran}
\usepackage{amsmath,amssymb,amsthm,amsfonts}
\usepackage{graphicx}
\usepackage{bm}
\usepackage{cite}
\usepackage{url}
\usepackage{enumitem}

\newtheorem{theorem}{Theorem}
\newtheorem{lemma}{Lemma}
\newtheorem{proposition}{Proposition}
\newtheorem{corollary}{Corollary}
\newtheorem{definition}{Definition}
\newtheorem{assumption}{Assumption}
\newtheorem{remark}{Remark}

\newcommand{\R}{\mathbb{R}}
\newcommand{\Z}{\mathbb{Z}}
\newcommand{\E}{\mathbb{E}}
\newcommand{\Prob}{\mathbb{P}}
\newcommand{\lat}[1]{\Lambda_{#1}}
\newcommand{\vor}[1]{\mathcal{V}_{#1}}
\newcommand{\md}[1]{\ \mathrm{mod}\ #1}
\newcommand{\rcov}[1]{r_{\mathrm{cov}}(#1)}
\newcommand{\rpack}[1]{r_{\mathrm{pack}}(#1)}

\begin{document}

\title{Progressive-Resolution Secure Aggregation for Federated Learning}

\author{Seyed Mohammad Azimi-Abarghouyi,~\IEEEmembership{Member,~IEEE}
\thanks{The author is with the Department of Electrical Engineering,
Chalmers University of Technology, Gothenburg, Sweden
(Email: azimimo@chalmers.se).}}

\maketitle

\begin{abstract}
Secure aggregation lets a server recover an aggregate of client updates
without observing any individual update, but conventional protocols fix the
aggregate precision when clients upload. We introduce and formulate a new progressive-resolution secure-aggregation functionality in which clients upload once and successively finer resolutions of the same aggregate can later be authorized without renewed client participation. To realize this functionality, we propose progressive-resolution secure aggregation (PSA):
each clipped, dithered update is represented by compatible nested-lattice
digits; separately releasable layers are protected by secure aggregation and
an additional aggregate pad that remains unavailable to the server until a
non-colluding release controller authorizes that layer. Separate modular sums
create an overflow risk, for which we derive worst-case and concentration-based
margin rules. Under independent zero-mean layer symbols and normalized
aggregate-noise scale $O(\sqrt K)$, the sufficient margin for $K$ clients and
radix $\rho$ is $\tfrac12\log_\rho K+O(1)$ radix steps. Under an i.i.d.\
uniform-digit prior, a matching converse shows that the same order is necessary
within the class of interfaces that expose each separately authorized release
through one modular sum; equivalently, each additional release carries an
arithmetic-payload premium of order $\tfrac12\log_2 K$ bits per model
dimension relative to a single-stage interface. We also give release-by-release
differential-privacy accounting and resolution-dependent distortion bounds.
Exact finite-alphabet calculations exhibit margin scaling consistent with the
analysis. Federated experiments on MNIST and CIFAR-10 show that PSA coincides
with a flat secure sum when progressive release is not used and quantify the
accuracy--payload trade-off against distributed discrete-Gaussian and Skellam
mechanisms at matched privacy.
\end{abstract}
\vspace{-10pt}
\vspace{-10pt}
\section{Introduction}
\label{sec:intro}

\IEEEPARstart{F}{ederated} learning (FL) keeps raw data at clients and exposes
model updates through aggregation \cite{kairouz_survey}. Secure aggregation
(SecAgg) cryptographically restricts the server to a prescribed aggregate,
while distributed differential privacy (DP) can limit what the released
aggregate reveals by adding client-side noise before aggregation
\cite{bonawitz,bell,ddgauss,skellam}.

These mechanisms control the content of an aggregate but not the precision at
which it is disclosed. Once a fine-precision aggregate has been decoded,
rounding it afterward cannot undo the finer disclosure. This matters when the
precision required by an approved analysis is not known when clients upload,
particularly in cross-device settings where participants may be unavailable
for a second upload. Committing every query to the finest potentially useful
precision can therefore disclose more than the eventual purpose requires.
For example, a coarse hospital statistic may already determine on which side
of a decision threshold the population lies, making a finer value unnecessary.
Disclosure-control guidance similarly cautions against publishing more
precision than a purpose requires \cite{ons_disclosure}, the data-minimization
principle of the EU General Data Protection Regulation limits processing to
what is necessary for the stated purpose \cite{gdpr}, and tiered data-access
agreements operationalize the same idea by granting an analyst the coarsest
view sufficient for an approved analysis. The resulting design
goal is a minimum-necessary-resolution policy: release a coarse aggregate
first and refine it only when justified, without recalling the clients.

Existing approaches do not provide this combination. Lattice quantizers for FL
provide fixed-resolution compression and may incorporate privacy noise
\cite{uveqfed,jopeq,cpa,olala}; threshold or split-trust systems can delay
aggregate decryption \cite{paillier,prio}; and distributed DP under SecAgg
controls the information released by an aggregate \cite{ddgauss,skellam}.
None of these mechanisms by itself gives mutually consistent coarse and fine
parts that can be authorized separately after one client upload.

To realize this interface, we propose \emph{progressive-resolution secure aggregation} (PSA). PSA combines two logically distinct components. First, a
compatible nested-lattice representation decomposes each quantized update into
bounded refinement layers such that every newly released layer refines the
same quantized value rather than replacing it by an independently quantized
surrogate. Second, a sealed-release gate adds data-independent pad shares whose
aggregate remains unknown to the server until a non-colluding controller or
trustee set authorizes the corresponding layer. Ordinary SecAgg still protects
individual uploads. The gate is treated as an abstract functionality rather than a new threshold
cryptosystem: its realization is delegated to standard threshold or multi-server
tools \cite{paillier,prio}, and the contribution of this paper is the
progressively refinable representation presented to that gate together with a
characterization of the costs induced by making its layers separately
releasable. Thus, PSA separates the algebraic mechanism that enables consistent coarse-to-fine reconstruction from the release mechanism that controls when each refinement becomes visible.

The structured representation is needed because independently releasable
pieces must remain closed under aggregation and admit controlled modular
reduction. A compatible lattice chain supplies both properties: layer sums
remain on the corresponding lattice, and successive released layers reconstruct
a consistent coarse-to-fine sequence. A gate without such a representation
would only delay an all-or-nothing value, while the representation without a
gate would not enforce delayed disclosure. To our knowledge, this is the first formulation of secure aggregation that enables post-upload progressive authorization of successively finer resolutions of the same aggregate without renewed client participation. We further characterize the overflow, communication, privacy, and distortion costs induced by making these refinements separately releasable.
\vspace{-10pt}
\subsection{Contributions}

The paper makes the following contributions.
\begin{enumerate}[leftmargin=1.4em]
\item We formalize progressive-resolution aggregate release after a single
client upload and give a PSA architecture based on compatible nested-lattice
layers, modular SecAgg, and aggregate release pads.
\item For $K$ clients and radix $\rho$, we characterize the carries induced by
separate modular reductions. We derive a worst-case-data margin of order
$\log_\rho K$ and, under independent zero-mean layer symbols with normalized
aggregate-noise scale $O(\sqrt K)$, a margin
$\tfrac12\log_\rho K+O(1)$.
\item For interfaces exposing each separately authorized release only through
one modular sum, we prove the same $\tfrac12\log_\rho K$ margin order necessary
under an i.i.d.\ uniform-digit prior. This converse is interface-specific and
is not a lower bound for arbitrary interactive or jointly encoded protocols.
\item We give release-by-release R\'enyi-DP composition, instantiate it with
product discrete-Gaussian noise, and compare its leading sensitivity cost with
a matched flat representation. When the admissible range fills the balanced
digit alphabet in every coordinate, the leading bound is neutral rather than
strictly better.
\item We derive exact full-resolution dithered-quantization mean-squared error
(MSE), distribution-free truncation bounds, and geometric truncation scaling
under an explicit uniform-residue condition.
\item We evaluate the claims in the regimes in which they are stated: exact
finite-alphabet calculations examine carry and arithmetic-payload scaling, while
federated experiments on MNIST and CIFAR-10 compare PSA with distributed
discrete-Gaussian and Skellam mechanisms at matched differential-privacy targets
and use a compression-only lattice baseline to isolate representation cost.
\end{enumerate}

\section{Related Work}
\label{sec:related}

\emph{Secure aggregation and distributed DP.} Pairwise-masked SecAgg reveals a
sum while hiding individual inputs \cite{bonawitz}, with later protocols
reducing its overhead \cite{bell}. The distributed discrete-Gaussian
\cite{ddgauss} and Skellam \cite{skellam} mechanisms add client-side integer
noise and account for quantization and modular wraparound of a single released
sum. PSA differs in functionality: several mutually consistent resolutions
must remain separately releasable after one upload, so each separately sealed
sum requires its own overflow protection.

\emph{Delayed decryption and split trust.} Additively homomorphic and threshold
encryption \cite{paillier}, and multi-server systems such as Prio \cite{prio},
can withhold decryption capability until an authorized party or quorum acts.
PSA can use such mechanisms to realize its release controller, but does not
claim a new decryption primitive. Its contribution is the object presented to
that gate: a coarse-to-fine decomposition whose authorized parts reconstruct
one consistent aggregate.

\emph{Quantization for FL.} UVeQFed \cite{uveqfed} uses subtractive-dithered
lattice quantization of model updates; JoPEQ \cite{jopeq} combines lattice
quantization with privacy noise; compressed private aggregation \cite{cpa}
uses nested quantizers for rate refinement within a release; and OLALa
\cite{olala} adapts lattices to heterogeneous clients. In these methods,
nesting primarily serves rate--distortion or adaptation. PSA instead assigns
nesting to an authorization order, which requires successive pieces to
reconstruct one another exactly rather than merely reduce distortion. ScionFL
\cite{scionfl} addresses malicious inputs to modular secure aggregation, a
threat excluded by the present honest-but-curious model and therefore
complementary to PSA.

\emph{Nested lattices and computation.} Nested lattices are established
structured-coding tools \cite{zamirbook}, including for multi-layer computation
\cite{aircpu} and lattice-coded FL \cite{fedcpu, fedcpu_conf}. Those works exploit closure
under integer combinations. PSA uses the same algebraic property for secure
aggregation, but additionally requires a digit-compatibility condition so that
successive quantizer differences are valid bounded representatives. The
canonical construction used here is the balanced odd-radix integer chain.
\vspace{-10pt}
\section{Preliminaries}
\label{sec:prelim}

\subsection{Lattice Notation and Compatible Digits}

For a block dimension $n_{\mathrm b}\ge1$, let a full-rank lattice
$\lat{}\subset\R^{n_{\mathrm b}}$ have nearest-neighbor quantizer
$Q_{\lat{}}(\mathbf{x})\in\arg\min_{\bm\lambda\in\lat{}}
\|\mathbf{x}-\bm\lambda\|$. Its centered Voronoi region is $\vor{}$, with
boundary ties fixed deterministically, and centered reduction is
$[\mathbf{x}]\md{\lat{}}=\mathbf{x}-Q_{\lat{}}(\mathbf{x})$. Denote its
covering radius, packing radius, and per-dimension second moment by
$\rcov{\lat{}}$, $\rpack{\lat{}}$, and $\sigma^2(\lat{})$, respectively.

Fix a number of layers $L\ge1$ and an odd integer radix $\rho\ge3$. Consider
the self-similar chain
\begin{align}\label{eq:chain}
\lat{\ell+1}=\rho\lat{\ell},\
\lat{L+1}\subset\cdots\subset\lat1,
\ \ell\in\{1,\ldots,L\}.
\end{align}
For each layer, the quotient group $G_\ell=\lat\ell/\lat{\ell+1}$ has
$|G_\ell|=\rho^{n_{\mathrm b}}$. Let
$\mathcal C_\ell\triangleq\lat\ell\cap\vor{\ell+1}$ under the same boundary
convention. Then $\mathcal C_\ell$ is a complete set of coset representatives
for $G_\ell$, so every element of $\lat\ell$ is congruent modulo
$\lat{\ell+1}$ to exactly one element of $\mathcal C_\ell$.

Nesting alone does not ensure that consecutive quantizers produce one of these
representatives. PSA therefore uses the following structural condition.
\begin{assumption}[Digit-compatible chain]\label{ass:compatible}
For every $\mathbf{x}\in\lat1\cap\vor{L+1}$ and every
$\ell\in\{1,\ldots,L\}$,
\begin{align}\label{eq:compat}
Q_{\lat\ell}(\mathbf{x})-Q_{\lat{\ell+1}}(\mathbf{x})
\in\mathcal C_\ell .
\end{align}
\end{assumption}

The canonical construction satisfying Assumption~\ref{ass:compatible} is
\begin{align}\label{eq:intchain}
\lat1=h\Z^{n_{\mathrm b}},\qquad
\lat\ell=h\rho^{\ell-1}\Z^{n_{\mathrm b}},
\end{align}
where $h>0$. Coordinatewise nearest-integer rounding gives the balanced
radix-$\rho$ alphabet
\begin{align}\label{eq:intalphabet}
\mathcal C_\ell=h\rho^{\ell-1}
\left\{-\frac{\rho-1}{2},\ldots,\frac{\rho-1}{2}\right\}^{n_{\mathrm b}},
\end{align}
and \eqref{eq:compat} follows from balanced Euclidean division. For this
integer chain with odd $\rho$, no point of $\lat\ell$ lies on the boundary of
$\vor{\ell+1}$, so the alphabet needs no boundary tie-breaking and satisfies
$\mathcal C_\ell=-\mathcal C_\ell$. Thus odd $\rho$ makes the representative
alphabet symmetric about the origin. If $\mathbf c$ is uniform on
$\mathcal C_\ell$, then
\begin{align}\label{eq:layerenergy}
\E\|\mathbf c\|^2
=n_{\mathrm b}h^2\rho^{2(\ell-1)}\frac{\rho^2-1}{12}.
\end{align}
Equation~\eqref{eq:layerenergy} is a property of a uniform digit model; it is
not implied by subtractive dithering or by Assumption~\ref{ass:compatible}.
\vspace{-10pt}
\section{Progressive-Resolution Secure Aggregation}
\label{sec:scheme}
\vspace{-4pt}
\subsection{Setting}

There are $K\ge1$ clients indexed by $k\in\{1,\ldots,K\}$. Client $k$ holds
an update $\mathbf u_k\in\R^d$, where $d$ is the model dimension. Let
$\bm\mu_0\in\R^d$ be a shared reference and let $\Delta>0$ be the clipping
radius. Writing $\operatorname{clip}_\Delta$ for clipping to Euclidean norm at
most $\Delta$, define
\begin{align}\label{eq:target}
\mathbf s_\Delta
=K\bm\mu_0+\sum_{k=1}^K\mathbf g_k,
\qquad
\mathbf g_k=\operatorname{clip}_\Delta(\mathbf u_k-\bm\mu_0),
\end{align}
so $\|\mathbf g_k\|\le\Delta$. The protocol targets
$\mathbf s_\Delta$; when clipping is inactive,
$\mathbf s_\Delta=\sum_k\mathbf u_k$. Each $d$-dimensional vector is processed
in $J=\lceil d/n_{\mathrm b}\rceil$ lattice blocks, zero-filled in the last
block if necessary. Recentering on $\bm\mu_0$, typically the previous global
model, reduces the range that the lattice chain must cover.
\vspace{-10pt}
\subsection{Threat Model and Scope}
\label{sec:threat}

For \emph{individual-update privacy}, PSA retains the honest-but-curious
SecAgg model: the aggregation server may collude with a bounded number of
clients, and dropouts are handled by the underlying SecAgg recovery protocol
\cite{bonawitz,bell}. For \emph{delayed aggregate release}, PSA additionally
assumes that the aggregation server does not collude before authorization with
the release controller or with a threshold of trustees. The controller handles
only data-independent gate material and receives neither individual updates nor
an unmasked aggregate.

Each separately authorizable layer will be sealed by an aggregate release pad.
The information-theoretic delayed-release statements require that pad, for every
unreleased layer and realized active set, to be uniform and unknown to the
server; pseudorandom pads replace perfect hiding by computational
indistinguishability. If the controller colludes with the server before a
release, delayed authorization for that release is lost, while the underlying
SecAgg protection of individual uploads remains. The analysis excludes active
poisoning, malformed lattice symbols, verifiable authorization policy, and
traffic analysis; participation metadata and the identities of released layers
are public. Robust input validation such as \cite{scionfl} is complementary;
because each PSA symbol is a bounded digit, such validation reduces to
per-layer range checks on the modular uploads, at a cost that scales with the
number of separately sealed sums.
The release pad and DP noise serve different purposes: the pad hides an
unauthorized layer aggregate, whereas DP limits what an authorized aggregate
reveals. The gate does not remove information about a withheld refinement that
is already implied by authorized outputs.
\vspace{-10pt}
\subsection{Client Encoding}
\label{sub:enc}

\emph{1) Scaling and subtractive dithering.} A public gain $a>0$ maps the
clipped update into the lattice range. The formulas below describe one
$n_{\mathrm b}$-dimensional block; to avoid an additional block index, that
block of $\mathbf g_k$ is denoted again by $\mathbf g_k$. Client $k$ sets
$\mathbf v_k=a\mathbf g_k$ and draws an independent public dither
$\mathbf d_k\sim\mathrm{Unif}(\vor1)$, independent of the update. Define
\begin{align*}
\mathbf e_k
=Q_{\lat1}(\mathbf v_k+\mathbf d_k)-\mathbf d_k-\mathbf v_k.
\end{align*}
Under the standard subtractive-dither model \cite{zamirbook}, $\mathbf e_k$
is independent of $\mathbf v_k$, has zero mean,
$\E\|\mathbf e_k\|^2=n_{\mathrm b}\sigma^2(\lat1)$, and satisfies
$\|\mathbf e_k\|\le\rcov{\lat1}$ almost surely. The client forms
$\mathbf x_k=Q_{\lat1}(\mathbf v_k+\mathbf d_k)\in\lat1$, hence
$\mathbf x_k=\mathbf v_k+\mathbf d_k+\mathbf e_k$ and
\begin{align*}
\|\mathbf x_k\|\le B\triangleq a\Delta+2\rcov{\lat1}.
\end{align*}
The sufficient range condition
$B<\rpack{\lat{L+1}}=\rho^L\rpack{\lat1}$ places every $\mathbf x_k$ inside
$\vor{L+1}$ and gives $Q_{\lat{L+1}}(\mathbf x_k)=\mathbf0$.

\emph{2) Compatible telescoping decomposition.} Under
Assumption~\ref{ass:compatible}, for each $\ell\in\{1,\ldots,L\}$ define
\begin{align}\label{eq:telescope}
\mathbf x_k^\ell
&=Q_{\lat\ell}(\mathbf x_k)-Q_{\lat{\ell+1}}(\mathbf x_k)
\in\mathcal C_\ell,\\
\mathbf x_k&=\sum_{\ell=1}^{L}\mathbf x_k^\ell .
\end{align}
For any $\ell_0\in\{1,\ldots,L\}$, telescoping also gives
\begin{align}\label{eq:partial}
\sum_{\ell=\ell_0}^{L}\mathbf x_k^\ell
=Q_{\lat{\ell_0}}(\mathbf x_k),
\end{align}
so larger layer indices are coarser and decreasing $\ell_0$ adds finer
refinements of the same quantized point.

\emph{3) Per-layer DP noise.} Client $k$ draws zero-mean lattice noise
$\mathbf z_k^\ell\in\lat\ell$, independently across clients and layers and
independently of the data and dithers, and transmits the noisy layer symbol
$\widetilde{\mathbf x}_k^\ell=\mathbf x_k^\ell+\mathbf z_k^\ell$. For the
integer-chain DP instantiation, each coordinate has probability mass
\begin{align}\label{eq:dgdef}
\Prob[Z=h\rho^{\ell-1}j]
\propto\exp\!\left[-\frac{(h\rho^{\ell-1}j)^2}{2s_\ell^2}\right],
\qquad j\in\Z,
\end{align}
where $s_\ell>0$ is the discrete-Gaussian scale in physical lattice units.
Its actual per-coordinate variance is denoted $\bar s_\ell^2$. This is the
discrete Gaussian analyzed in \cite{canonne}.
\vspace{-10pt}
\subsection{Per-Layer Secure Aggregation and Sealed Release}
\label{sub:gate}

A single-client digit lies in $\mathcal C_\ell$, whereas the sum of $K$ such
digits has a larger dynamic range. For each layer choose an integer overflow
margin $m_\ell\ge0$ and define the modulus lattice and finite quotient
\begin{align}\label{eq:modgroup}
\lat{}^{(\ell)}&=\rho^{m_\ell}\lat{\ell+1},\\
\bar G_\ell&=\lat\ell/\lat{}^{(\ell)}
\cong(\Z/\rho^{1+m_\ell}\Z)^{n_{\mathrm b}}.
\end{align}
Let $\mathcal A_\ell\subseteq\{1,\ldots,K\}$ denote the realized active set
for layer $\ell$. Client $k\in\mathcal A_\ell$ contributes a release-pad share
$\mathbf b_k^\ell\in\bar G_\ell$. The shares are independent across clients,
layers, and rounds. If at least one active honest client contributes a share
that is uniform on $\bar G_\ell$, independent of the remaining shares, and
unknown to the server, then the aggregate pad
\begin{align}\label{eq:pad}
\mathbf B_\ell=\left[\sum_{k\in\mathcal A_\ell}\mathbf b_k^\ell\right]
\md{\lat{}^{(\ell)}}
\end{align}
is uniform on $\bar G_\ell$ from the server's view.

The release controller is assumed able to recover $\mathbf B_\ell$ for the
realized active set without renewed client participation, for example from
per-client keys or trustee-held shares. This capability is assumed rather than
constructed. With true one-time pad shares the resulting ciphertext hiding is
information-theoretic; pseudorandom-generator (PRG) or key-derivation-function
(KDF) shares provide the corresponding computational property and require
domain separation across layers and rounds.

Let $\mathbf m_{kj}^\ell\in\bar G_\ell$ denote the standard pairwise SecAgg
mask shared by clients $k$ and $j$, with
$\mathbf m_{kj}^\ell=-\mathbf m_{jk}^\ell$. Client $k$ submits
\begin{align}\label{eq:upload}
\mathbf y_k^\ell=
\Big[\widetilde{\mathbf x}_k^\ell+\mathbf b_k^\ell
+\sum_{j\ne k}\mathbf m_{kj}^\ell\Big]
\md{\lat{}^{(\ell)}} .
\end{align}
After the ordinary SecAgg masks cancel over the active set, the server holds
only the sealed aggregate
\begin{align}\label{eq:sealedagg}
\mathbf u_\ell
=\Big[\sum_{k\in\mathcal A_\ell}\widetilde{\mathbf x}_k^\ell
+\mathbf B_\ell\Big]\md{\lat{}^{(\ell)}}.
\end{align}
When layer $\ell$ is authorized, the controller supplies $\mathbf B_\ell$ and
the server obtains
\begin{align}\label{eq:tell}
\mathbf t_\ell
=\big[\mathbf u_\ell-\mathbf B_\ell\big]\md{\lat{}^{(\ell)}}
=\Big[\sum_{k\in\mathcal A_\ell}\widetilde{\mathbf x}_k^\ell\Big]
\md{\lat{}^{(\ell)}}.
\end{align}
Thus SecAgg masks protect individual uploads and cancel at aggregation, whereas
the aggregate release pad remains until authorization.

For the homogeneous active set $\mathcal A_\ell=\{1,\ldots,K\}$ at every
released layer, let $\ell_0\in\{1,\ldots,L\}$ be the finest authorized layer,
so the authorized suffix is $\{\ell_0,\ldots,L\}$. Using
\eqref{eq:partial}, the server reconstructs
\begin{align}\label{eq:recon}
\widehat{\mathbf s}_{\Delta,\ell_0}
=K\bm\mu_0+\frac1a\left(
\sum_{\ell=\ell_0}^{L}\mathbf t_\ell
-\sum_{k=1}^{K}\mathbf d_k\right).
\end{align}
Authorizing one finer layer requires one additional controller token but no new
client upload. An optional common orthogonal rotation may be applied to the
clipped updates before block quantization to spread update energy across
coordinates and reduce the per-coordinate dynamic range; it is inverted inside
the parentheses of \eqref{eq:recon} and changes neither the subsequent algebra
nor the privacy accounting.

\begin{remark}[Authorization schedules]\label{rem:schedule}
The per-layer construction gives the finest authorization granularity. For the
canonical integer chain \eqref{eq:intchain}, contiguous layers can instead be
grouped into $R\in\{1,\ldots,L\}$ stages. Index the stages in fine-to-coarse
arithmetic order as an ordered partition
$\mathcal G_1,\ldots,\mathcal G_R$ of $\{1,\ldots,L\}$, with each
$\mathcal G_r$ contiguous; policy releases them in the reverse order
$\mathcal G_R,\ldots,\mathcal G_1$. For $r\in\{1,\ldots,R\}$, let
$b_r=|\mathcal G_r|$, $\ell_r^-=\min\mathcal G_r$, and
$\ell_r^+=\max\mathcal G_r$. Define
$\Lambda'_r\triangleq\lat{\ell_r^-}$ for $r\in\{1,\ldots,R\}$ and
$\Lambda'_{R+1}\triangleq\lat{L+1}$. These lattices form the coarsened
integer chain
$\Lambda'_{r+1}=\rho^{b_r}\Lambda'_r$ for $r\in\{1,\ldots,R\}$,
with $\Lambda'_1=\lat1$.
Balanced Euclidean division therefore gives a compatible stage symbol
\begin{align*}
\mathbf x_k^{r}
=\sum_{\ell\in\mathcal G_r}\mathbf x_k^\ell
=Q_{\lat{\ell_r^-}}(\mathbf x_k)
-Q_{\lat{\ell_r^++1}}(\mathbf x_k)
\in\lat{\ell_r^-}\cap\vor{\ell_r^++1}.
\end{align*}
A stage uses one pad and one independent noise draw
$\mathbf z_k^r\in\lat{\ell_r^-}$, with actual per-coordinate variance
$\bar s_r^2$, and a modulus
$\lat{}^{(r)}=\rho^{m_r}\lat{\ell_r^++1}$. The corresponding quotient has
$\rho^{n_{\mathrm b}(b_r+m_r)}$ elements. The per-layer statements apply to
a stage through this coarsened integer chain. Expressed in the
original radix, a concentration margin of
$\tfrac12\log_{\rho^{b_r}}K+O(1)$ coarsened-chain steps corresponds to
$b_r$ times as many original-radix exponents, hence to
$\tfrac12\log_\rho K+O(1)$ original-radix margin steps. The extremes are $R=L$ (one
layer per stage) and $R=1$ (one flat stage).
\end{remark}

\begin{remark}[Relation to a flat secure sum]\label{rem:flat}
With one layer, or with one stage containing all $L$ layers, PSA reduces to one
quantized modular sum plus an optional release gate. If immediate disclosure is
allowed, setting the release pads to zero recovers ordinary SecAgg arithmetic.
\end{remark}
\vspace{-14pt}
\subsection{Secrecy of an Unreleased Layer}

\begin{proposition}[Ciphertext hiding for an unreleased layer]
\label{prop:secrecy}
Fix a layer $\ell$ and condition on its active set and all public randomness.
Let $\mathbf W_\ell$ denote the true noisy modular sum and let $V$ denote any
other server-side information that may be correlated with $\mathbf W_\ell$.
Assume, as in the pad-uniformity condition of Section~\ref{sub:gate}, that
$\mathbf B_\ell$ is uniform and independent of $(\mathbf W_\ell,V)$ from
the server's view. Define the sealed random variable
$\mathbf U_\ell=[\mathbf W_\ell+\mathbf B_\ell]\md{\lat{}^{(\ell)}}$.
Then
$I(\mathbf W_\ell;\mathbf U_\ell\mid V)=0$, where $I(\cdot;\cdot\mid\cdot)$
denotes conditional mutual information. Thus the sealed ciphertext adds no
information about the withheld layer beyond the server's pre-existing side
information. After $\mathbf B_\ell$ is released, the server learns exactly
$\mathbf W_\ell$. With PRG-generated pads, the corresponding statement is
computational rather than information-theoretic.
\end{proposition}

\begin{proof}
For every fixed pair $(\mathbf w,v)$, addition of an independent uniform pad
maps $\mathbf w$ to a uniform element of $\bar G_\ell$. Hence
$\mathbf U_\ell$ is conditionally uniform and independent of
$\mathbf W_\ell$ given $V$. Standard SecAgg separately protects the individual
uploads.
\end{proof}
\vspace{-10pt}
\section{Exactness, Carries, and Overload}
\label{sec:carry}

After authorization, a layer still provides only a centered modular residue.
To characterize when that residue equals the true noisy sum, assume in this
section that every client contributes to every layer and define
$\mathbf S_\ell=\sum_{k=1}^K\mathbf x_k^\ell$ and
$\mathbf Z_\ell=\sum_{k=1}^K\mathbf z_k^\ell$.

\begin{lemma}[Carry structure]\label{lem:carry}
For each released layer there is a unique
\begin{align}
\mathbf c_\ell
=Q_{\lat{}^{(\ell)}}(\mathbf S_\ell+\mathbf Z_\ell)
\in\lat{}^{(\ell)}
\end{align}
such that
$\mathbf S_\ell+\mathbf Z_\ell=\mathbf t_\ell+\mathbf c_\ell$.
The carry is zero if and only if $\mathbf S_\ell+\mathbf Z_\ell$ lies in
the centered Voronoi region of $\lat{}^{(\ell)}$. If it is
nonzero, then
$\|\mathbf c_\ell\|\ge2\rpack{\lat{}^{(\ell)}}$.
\end{lemma}

\begin{proof}
This is the definition of centered reduction modulo
$\lat{}^{(\ell)}$; a nonzero point of that lattice has norm at
least its minimum distance, $2\rpack{\lat{}^{(\ell)}}$.
\end{proof}

Carries also localize by scale. If $\mathcal B$ is the set of layers with
nonzero carries, their contribution to the reconstruction error at resolution
$\ell_0$ is
$-a^{-1}\sum_{\ell\in\mathcal B,\,\ell\ge\ell_0}\mathbf c_\ell$;
a carry on an unreleased finer layer $\ell<\ell_0$ cannot alter that coarser
authorized output. This localization does not make a carry on a coarse layer
small.

A nonzero carry is therefore not a small perturbation; overflow margins are
designed to make it absent with high probability. Because the DP noise is
unbounded, zero carry is necessarily probabilistic when DP is enabled. We use
$\delta_w\in(0,1)$ for concentration failure of the data term and
$\delta_z\in(0,1)$ for the noise tail. We therefore give a worst-case rule that makes no distributional assumption on
layer symbols and a smaller-margin rule under an explicit zero-mean and
independence assumption.

\begin{assumption}[Zero-mean layer symbols]\label{ass:qu}
For every layer on which the concentration rule is used, the
vectors $\{\mathbf x_k^\ell\}_{k=1}^K$ are independent across
clients and satisfy $\E[\mathbf x_k^\ell]=\mathbf0$.
\end{assumption}

Subtractive dithering makes the quantization error unbiased but does not imply
Assumption~\ref{ass:qu}. For the integer chain, the assumption holds under the
i.i.d.\ uniform balanced-digit model used in the converse; its suitability for
federated updates is an empirical workload question.

For both margin rules, the left-hand side is
$\rho^{m_\ell}p_{\ell+1}=\rpack{\lat{}^{(\ell)}}$: a point of smaller norm
lies in the zero cell, so its carry vanishes by Lemma~\ref{lem:carry}. Each
rule bounds $\|\mathbf S_\ell+\mathbf Z_\ell\|$ and makes the margin exceed
it.

\begin{theorem}[Margin rules]\label{thm:margin}
Fix $\delta_w,\delta_z\in(0,1)$. Let $r_{\ell+1}=\rcov{\lat{\ell+1}}$,
$p_{\ell+1}=\rpack{\lat{\ell+1}}$, and choose any $u_\ell$ such that
$\Prob[\|\mathbf Z_\ell\|>u_\ell]\le\delta_z/L$.
For the product discrete Gaussian on the integer chain, one valid choice is
\begin{align}\label{eq:uint}
u_\ell=s_\ell\sqrt{2K n_{\mathrm b}
\log\frac{2n_{\mathrm b}L}{\delta_z}},
\end{align}
valid whenever
$\widetilde s_\ell\triangleq s_\ell/(h\rho^{\ell-1})\ge1/2$. Under this
condition each coordinate of
$\mathbf z_k^\ell$ is sub-Gaussian with proxy $s_\ell$ \cite{mp13}, so each
coordinate of $\mathbf Z_\ell$ has proxy $s_\ell\sqrt K$, and
\eqref{eq:uint} follows by a union bound over coordinates. Then:
\begin{enumerate}[leftmargin=1.4em,label=(\alph*)]
\item \emph{Worst-case data.} If
\begin{align}\label{eq:worstmargin}
\rho^{m_\ell}p_{\ell+1}>K r_{\ell+1}+u_\ell,
\end{align}
then $\Prob[\mathbf c_\ell\ne\mathbf0]\le\delta_z/L$.
\item \emph{Concentrated data.} Under
Assumption~\ref{ass:qu}, define the total layer second moment
$V_\ell=\sum_{k}\E\|\mathbf x_k^\ell\|^2$.
If
\begin{align}\label{eq:conc}
\rho^{m_\ell}p_{\ell+1}>
&\sqrt{2V_\ell\log\frac{4L}{\delta_w}}
+\frac23r_{\ell+1}\log\frac{4L}{\delta_w}
+u_\ell,
\end{align}
then
$\Prob[\mathbf c_\ell\ne\mathbf0]
\le(\delta_w+\delta_z)/L$.
\end{enumerate}
If part (b) is used on all layers, a union bound gives total carry probability
at most $\delta_w+\delta_z$; if part (a) is used on all layers, the
corresponding bound is $\delta_z$. If
$V_\ell^{1/2}/(p_{\ell+1}\sqrt K)$ and
$u_\ell/(p_{\ell+1}\sqrt K)$ are bounded uniformly in $K$ and $\ell$, part
(b) yields $m_\ell=\tfrac12\log_\rho K+O(1)$; part (a) yields
$m_\ell=\log_\rho K+O(1)$ under the analogous $K$-scaling, its data term
growing as $K$ rather than $\sqrt K$. For i.i.d.\ client layers, $V_\ell=K\bar\varsigma_\ell^2$ with
$\bar\varsigma_\ell^2=\E\|\mathbf x_k^\ell\|^2$; under the uniform
balanced-digit model, \eqref{eq:layerenergy} gives this second moment
explicitly.
\end{theorem}

\begin{proof}
Appendix~\ref{app:margin}.
\end{proof}
\vspace{-10pt}
\section{Communication Accounting}
\label{sec:rate}

Preserving the option of later refinement requires every separately sealed
stage to be represented in the initial upload. We first count the resulting
arithmetic payload and then give a converse for the modular interface used by
PSA. The count below is not total cryptographic traffic: it excludes SecAgg
key exchange, authentication, dropout-recovery communication, and any
controller/trustee key-management traffic.

\begin{proposition}[Client upload rate]\label{prop:rate}
Fix an authorization schedule of $R$ stages $\{\mathcal G_r\}$ as in
Remark~\ref{rem:schedule}. A protocol run that preserves the option of
later refinement uploads one message per stage, each an element of that
stage's extended quotient, so per model dimension the arithmetic payload
is
\begin{align}\label{eq:Rpsa}
\mathcal R_{\mathrm{PSA}}=\sum_{r=1}^{R}\bigl(|\mathcal G_r|+m_r\bigr)\log_2\rho
=\Bigl(L+\sum_{r=1}^{R}m_r\Bigr)\log_2\rho,
\end{align}
since the stages partition the $L$ layers. The payload therefore scales
with the number of separately sealed \emph{moduli}, not with the number of
layers: the cost is driven by refinement granularity. Let $\mathcal R_{\mathrm{flat}}$ denote the corresponding one-stage flat
payload. Before its single overflow margin, that representation contributes
$L\log_2\rho$. Margins are exponents of $\rho$, so a margin $m$ costs
$m\log_2\rho$ bits. With worst-case provisioning
Theorem~\ref{thm:margin}(a) gives
$m_r=\log_\rho K+O(1)$ for every modulus, that is $\log_2K+O(1)$ bits
each, and hence the gap $(R-1)\log_2K+O(R)$. Under the zero-mean
concentration assumptions, with aggregate-noise scales $O(\sqrt K)$ in
normalized units, Theorem~\ref{thm:margin}(b) gives
$m_r=\tfrac12\log_\rho K+O(1)$, that is $\tfrac12\log_2K+O(1)$ bits each,
and
\begin{align}\label{eq:gap}
\mathcal R_{\mathrm{PSA}}-\mathcal R_{\mathrm{flat}}
=\frac{R-1}{2}\log_2K+O(R).
\end{align}
The finest schedule $R=L$ recovers $\tfrac{L-1}{2}\log_2K+O(L)$, and $R=1$
gives no gap. These are asymptotics in $K$ for fixed $\rho$ and $L$; if $L$
grows, the $\log L$ terms from the failure allocation must be retained.
\end{proposition}

\begin{proof}
A stage $\mathcal G_r$ takes $\rho^{n_{\mathrm b}|\mathcal G_r|}$ values, so
its extended quotient element has
$n_{\mathrm b}(|\mathcal G_r|+m_r)\log_2\rho$ bits, against
$n_{\mathrm b}L\log_2\rho$ for the unextended fine index. The margin orders
follow from Theorem~\ref{thm:margin}, applied per stage and, for the flat
index, to its single modular sum; dividing by $n_{\mathrm b}$ gives the
per-dimension rates.
\end{proof}

The penalty is the price of protecting every layer against its own
wraparound within this arithmetic interface, not a universal lower bound:
interactive coding or a joint cross-layer encoding exposes a different
sufficient statistic. A converse must therefore fix an interface before
claiming necessity. The next result fixes the one PSA uses, in which each
authorized stage is exposed through a single modular sum, and shows the
order is unavoidable within it.
\vspace{-10pt}
\subsection{Converse for Layerwise Modular Interfaces}
\label{sec:converse}

\begin{definition}[Layerwise modular interface]\label{def:sealed}
Fix one coordinate $j$ of the balanced integer chain and normalize its layer
digit as
$a_k^\ell\triangleq x_{k,j}^\ell/(h\rho^{\ell-1})$, so
$a_k^\ell\in\{-(\rho-1)/2,\ldots,(\rho-1)/2\}$. Let
$S_\ell=\sum_{k=1}^K a_k^\ell$ and $M_\ell=\rho^{1+m_\ell}$. A decoder has a
layerwise modular interface if its only data-dependent observation from layer
$\ell$, after that layer is authorized, is
$t_\ell=S_\ell\bmod M_\ell$, together with arbitrary data-independent
randomness. This includes PSA's masks and release pads, which are
data-independent from the decoder's viewpoint.
\end{definition}

Under a uniform digit prior, $S_\ell$ has standard deviation of order
$\sqrt K$ and no single value has probability larger than order $K^{-1/2}$.
Consequently a residue modulo $M_\ell$ can identify the exact sum with
constant probability only if the modulus itself grows on the order of
$\sqrt K$.

\begin{theorem}[Anti-concentration converse]\label{thm:converse}
Under the i.i.d.\ uniform balanced-digit prior, let the normalized digits
$a_k^\ell$ of Definition~\ref{def:sealed} be independent across clients and
layers, and define
\begin{align}
S_\ell=\sum_{k=1}^K a_k^\ell,
\qquad
X=\sum_{\ell=1}^{L}\rho^{\ell-1}S_\ell.
\end{align}
Suppose the data-dependent observation from layer $\ell$ is only
$t_\ell=S_\ell\bmod M_\ell$, where
$M_\ell=\rho^{1+m_\ell}$. There are constants $C$ and $K_0$, depending at
most on the fixed radix, such that every decoder of $X$ from the modular
tuple, with estimate $\widehat X$, satisfies, for $K\ge K_0$ and every
$\ell$,
\begin{align}\label{eq:convbound}
\Prob[\widehat X=X]
\le \min\!\left\{1,\frac{C M_\ell}{\sigma_\rho\sqrt K}\right\},
\qquad
\sigma_\rho^2=\frac{\rho^2-1}{12}.
\end{align}
Therefore, for any $\delta\in(0,1)$, exact recovery with probability at least
$1-\delta$ requires
\begin{align}\label{eq:convmargin}
m_\ell\ge\frac12\log_\rho K
+\log_\rho\frac{(1-\delta)\sigma_\rho}{C\rho}
\end{align}
for every layer.
\end{theorem}

\begin{proof}
Appendix~\ref{app:converse}.
\end{proof}

\begin{corollary}[Rate within the layerwise modular class]
\label{cor:price}
Fix $L$ and $\rho$, and consider a one-symbol-per-client interface exposing
one modular sum per stage. A stage of width $b_r$ has per-coordinate
modulus $\rho^{b_r+m_r}$, so Theorem~\ref{thm:converse} applied to the
coarsened chain of Remark~\ref{rem:schedule} gives
$b_r+m_r\ge\tfrac12\log_\rho K+O(1)$, hence
$m_r\ge\tfrac12\log_\rho K+O(1)$ since $b_r\le L$ is fixed. Under the
uniform-digit prior and constant target success probability, the minimal
per-dimension rate is
$\mathcal R_{\mathrm{mod}}^\star(R)=L\log_2\rho+\tfrac R2\log_2K+O(R)$.
Relative to the same interface with a single stage,
\begin{align}
\mathcal R_{\mathrm{mod}}^\star(R)-\mathcal R_{\mathrm{mod}}^\star(1)
=\frac{R-1}{2}\log_2K+O(R),
\end{align}
which matches \eqref{eq:gap} in order. 
\end{corollary}

Proposition~\ref{prop:secrecy} and Corollary~\ref{cor:price} together isolate
the role of the overhead. Separate pads enforce coarse-only authorization
without recontacting clients; under the uniform-digit prior and within the
layerwise modular class of Definition~\ref{def:sealed}, the associated
arithmetic-payload order is unavoidable. No claim is made for protocols that
jointly encode layers, interact across them, or expose a richer sufficient
statistic.
\vspace{-10pt}
\section{Differential Privacy of Authorized Releases}
\label{sec:priv}

The release gate controls when a stage becomes visible; DP
controls what an authorized release reveals. We use client-level replacement
adjacency. For $q>1$, a mechanism is $(q,\varepsilon)$-R\'enyi differential
private (RDP) if the order-$q$ R\'enyi divergence between its outputs on every
pair of neighboring datasets is at most $\varepsilon$ \cite{renyi}. RDP
composes additively at a fixed order and is preserved by post-processing.
\vspace{-10pt}
\subsection{Layer Sensitivity}
Condition on the public dithers. Let $\Phi_\ell$ map one admissible client
update to its concatenated $d$-dimensional layer-$\ell$ symbol across the
$J=\lceil d/n_{\mathrm b}\rceil$ blocks. Its full-vector replacement
sensitivity is
$\mathcal D_\ell=\sup_{\mathbf u,\mathbf u'}
\|\Phi_\ell(\mathbf u)-\Phi_\ell(\mathbf u')\|_2$, where the supremum is over
two admissible client updates; let $D_\ell$ denote the corresponding
single-block sensitivity. Since $\mathbf x^\ell\in\mathcal C_\ell$,
\begin{align}\label{eq:sens}
D_\ell\le\operatorname{diam}(\mathcal C_\ell)
\le2\rcov{\lat{\ell+1}},
\qquad
\mathcal D_\ell\le\sqrt{J}\,D_\ell.
\end{align}
The $\sqrt J$ factor is a worst-case product bound. Under global $\ell_2$
clipping, the admissible set can make the actual full-vector sensitivity
smaller, in which case that smaller sensitivity should be used.

\begin{remark}[Layer sensitivities do not telescope]\label{rem:notele}
It is tempting to bound $\sum_\ell D_\ell$ by the sensitivity of the
total, since $\sum_\ell\mathbf x^\ell=\mathbf x$. This fails, and for
a concrete reason: an arbitrarily small change of $\mathbf x$ can
alter every layer symbol at once, exactly as incrementing $0999$ to
$1000$ changes every decimal digit. Layer sensitivities must
therefore be summed rather than telescoped, which is why the
composition cost of separately authorized layers does not vanish as
the layers become fine.
\end{remark}

For the integer chain \eqref{eq:intchain}, whenever the admissible
block domain can attain the full balanced digit alphabet,
$D_\ell=h(\rho-1)\rho^{\ell-1}\sqrt{n_{\mathrm b}}$.
If $B<\rpack{\lat\ell}$, then both
$Q_{\lat\ell}(\mathbf x)$ and
$Q_{\lat{\ell+1}}(\mathbf x)$ are zero for every admissible
$\mathbf x$, so the \emph{effective} sensitivity of that layer is
zero. We therefore omit inactive layers from the privacy sum.
\vspace{-10pt}
\subsection{Composition Across Authorized Layers}\label{sec:dp}

For an authorized set, only the mechanisms whose tokens are released enter the
privacy composition.

\begin{theorem}[Progressive RDP accounting]\label{thm:dp}
Condition on all public randomness and the realized active sets. Suppose
that, under client-level replacement adjacency, the authorized modular
mechanism for layer $\ell$ is
$(q,\varepsilon_\ell(q))$-RDP. Then every authorized set
$\mathcal P$ is
\begin{align}\label{eq:rdp}
\left(q,\ \varepsilon_{\mathcal P}(q)\right)\text{-RDP},
\qquad
\varepsilon_{\mathcal P}(q)
\le\sum_{\ell\in\mathcal P}\varepsilon_\ell(q).
\end{align}
Before its token is released, a uniformly sealed layer contributes no
additional leakage to the server transcript in the sense of
Proposition~\ref{prop:secrecy}. Modular reduction, unsealing by a
data-independent token, and affine reconstruction are post-processing.
Under an authorization schedule the mechanisms are the stages rather than
the layers, and the statement applies verbatim to the coarsened chain of
Remark~\ref{rem:schedule}, with $\mathcal P$ a set of stages.
\end{theorem}

\begin{proof}
RDP composes additively, including for schedules selected from public
information or prior DP outputs, and is preserved by post-processing. A
release decision based directly on additional private information would need
its own privacy analysis. The sealed transcript of every unreleased
layer is conditionally independent of its noisy sum by
Proposition~\ref{prop:secrecy}.
\end{proof}

For a concrete integer-chain instantiation, all blocks are concatenated and
each honest client adds product discrete Gaussian noise as in
\eqref{eq:dgdef}. Only honest clients contribute noise the accounting may rely on, a
colluding or dropped client's share being known to the server or absent;
let $K_{\mathrm{eff},\ell}$ lower-bound their number, set
$\widetilde s_\ell=s_\ell/(h\rho^{\ell-1})\ge1/2$. A sum of discrete
Gaussians is not exactly a discrete Gaussian; the following term
accounts for the residual. Define the convolution correction
\begin{align}\label{eq:tau}
\tau_\ell
=10\sum_{j=1}^{K_{\mathrm{eff},\ell}-1}
\exp\!\left(-2\pi^2\widetilde s_\ell^2\frac{j}{j+1}\right).
\end{align}
After rescaling the layer to $\Z^d$, the distributed-discrete-Gaussian
analysis \cite{ddgauss} supplies an explicit per-layer RDP bound. A
conservative form of its $\ell_2$ branch, sufficient for the comparisons
below, is
\begin{align}\label{eq:ddgrdp}
\varepsilon_\ell(q)
\le
\frac{q\mathcal D_\ell^2}
{2K_{\mathrm{eff},\ell}s_\ell^2}
+q\tau_\ell d.
\end{align}
Replacing \eqref{eq:ddgrdp} with the full minimum-of-branches RDP bound from \cite{ddgauss} gives a sharper numerical accountant; the enlarged
coefficient on the exponentially small correction keeps \eqref{eq:ddgrdp}
independent of concentrated-DP normalization conventions.

For a chosen DP failure probability $\delta\in(0,1)$, reported
$(\varepsilon,\delta)$ values are obtained from
$\varepsilon_{\mathcal P}(q)$ by the standard R\'enyi conversion, minimized
over the order $q>1$, and the noise scales are calibrated by bisection on
that minimized $\varepsilon$ at fixed $\delta$. The same RDP-to-$(\varepsilon,\delta)$ conversion and calibration criterion
are applied to the private mechanisms in the numerical evaluation, with each
mechanism using its stated RDP bound. Thus the reported comparisons are at
matched $(\varepsilon,\delta)$ rather than at matched R\'enyi order.

No carry-margin event is required for privacy. The mechanism always
releases a finite-group value, and any centered unwrapping rule is
post-processing. Carries affect utility, not the validity of the RDP
accounting.

\begin{remark}[General lattices]\label{rem:general}
Equation~\eqref{eq:ddgrdp} uses a product-integer discrete Gaussian.
Extending its constants to a nonproduct lattice such as $A_2$, $D_4$, or $E_8$ requires a corresponding convolution analysis
and a noise scale above the lattice smoothing parameter \cite{smoothing};
self-similarity alone is insufficient.
The geometric carry results, in contrast, apply to every chain satisfying
Assumption~\ref{ass:compatible}.
\end{remark}
\vspace{-20pt}
\subsection{Leading Privacy Cost of Layering}

By Theorem~\ref{thm:dp}, separately authorizable layers compose. To isolate the
leading sensitivity term, ignore the exponentially small $\tau_\ell$
corrections, suppose the same $K_{\mathrm{eff}}$ applies to all active layers,
and fix the nominal per-coordinate noise-scale budget
$V\triangleq\sum_\ell s_\ell^2$. This is an accounting comparison at fixed
$V$; because a discrete Gaussian's actual variance is $\bar s_\ell^2$ rather
than exactly $s_\ell^2$, it should not be read as an exact equal-distortion
statement at finite parameters.

\begin{proposition}[Leading composition cost at full release]
\label{prop:cost}
Let $\mathcal D_\ell$ be the actual full-vector replacement
sensitivity of layer $\ell$. Each layer contributes
$\mathcal D_\ell^2/s_\ell^2$, so the budget should be spent where a client
can move the output furthest. The allocation minimizing the leading term of
\eqref{eq:ddgrdp} is
\begin{align}\label{eq:optalloc}
s_\ell^2
=\frac{V\mathcal D_\ell}{\sum_j\mathcal D_j},
\end{align}
and yields
\begin{align}\label{eq:epslay}
\varepsilon_{\mathrm{lay}}^{\mathrm{lead}}(q)
=\frac{q(\sum_\ell\mathcal D_\ell)^2}
{2K_{\mathrm{eff}}V}.
\end{align}
If a matched flat representation has replacement sensitivity
$\mathcal D_{\mathrm{flat}}$ and the same noise-scale budget, then
\begin{align}\label{eq:ratio}
\frac{\varepsilon_{\mathrm{lay}}^{\mathrm{lead}}}
{\varepsilon_{\mathrm{flat}}^{\mathrm{lead}}}
=\left(\frac{\sum_\ell\mathcal D_\ell}
{\mathcal D_{\mathrm{flat}}}\right)^2.
\end{align}
For one full-range integer block, or a product of such blocks under
coordinate-wise range constraints,
\begin{align}
\mathcal D_\ell&=h(\rho-1)\rho^{\ell-1}\sqrt d,\\
\mathcal D_{\mathrm{flat}}&=h(\rho^L-1)\sqrt d,
\end{align}
so $\sum_\ell\mathcal D_\ell=\mathcal D_{\mathrm{flat}}$ and the
leading ratio is exactly $1$. More generally, if
$\mathcal D_\ell\le c_1\rho^{\ell-1}$ and
$\mathcal D_{\mathrm{flat}}\ge c_2\rho^{L-1}$ with positive
constants independent of $L$, then
\begin{align}\label{eq:generalratio}
\frac{\varepsilon_{\mathrm{lay}}^{\mathrm{lead}}}
{\varepsilon_{\mathrm{flat}}^{\mathrm{lead}}}
\le\left(\frac{c_1}{c_2}\frac{\rho}{\rho-1}\right)^2,
\end{align}
which is independent of $L$.
\end{proposition}

\begin{proof}
Minimizing $\sum_\ell\mathcal D_\ell^2/s_\ell^2$ at fixed
$\sum_\ell s_\ell^2=V$ is Cauchy--Schwarz: $(\sum_\ell\mathcal D_\ell)^2
\le(\sum_\ell\mathcal D_\ell^2/s_\ell^2)\sum_\ell s_\ell^2$, with equality
if and only if $s_\ell^2\propto\mathcal D_\ell$, giving
\eqref{eq:optalloc}--\eqref{eq:ratio}. For the full product domain,
$\sum_\ell\mathcal D_\ell=h(\rho-1)\sqrt d\sum_\ell\rho^{\ell-1}
=h(\rho^L-1)\sqrt d=\mathcal D_{\mathrm{flat}}$. In general geometric
growth gives $\sum_\ell\mathcal D_\ell\le c_1\rho^L/(\rho-1)$ while
$\mathcal D_{\mathrm{flat}}\ge c_2\rho^{L-1}$; dividing and squaring yields
\eqref{eq:generalratio}.
\end{proof}

Layering is therefore not automatically a privacy improvement: in the
tightly matched balanced integer case the leading \emph{composition upper
bound} is neutral under this budget comparison. The constraint
$\widetilde s_\ell\ge1/2$, the convolution corrections, and actual output
variances can shift the finite-parameter optimum. The caveat concerns the admissible domain, not the adjacency relation, which is
client-level replacement in both cases. The neutral case assumes a product box
that fills the digit alphabet coordinatewise. Under global $\ell_2$ clipping,
the admissible set is a ball and the layer sensitivities need not sum to the
flat sensitivity. Proposition~\ref{prop:cost} therefore identifies a regime in
which the leading composition bound is neutral; it does not claim neutrality
for general clipped workloads.
\vspace{-10pt}
\section{Distortion Analysis}
\label{sec:mse}

Fix a finest authorized layer $\ell_0\in\{1,\ldots,L\}$. Under a grouped
authorization schedule, $\ell_0$ is restricted to the minimum layer index of
the finest authorized stage. For a quantized block
$\mathbf x\in\lat1\cap\vor{L+1}$, define the refinement withheld below
$\ell_0$ by
\begin{align}\label{eq:resdef}
\mathbf r^{(\ell_0)}
\triangleq\mathbf x-Q_{\lat{\ell_0}}(\mathbf x)
\in\lat1\cap\vor{\ell_0},
\end{align}
and, under the actual data-and-dither law, let
$T_{\ell_0}\triangleq\E\|\mathbf r^{(\ell_0)}\|^2$. Always
$T_{\ell_0}\le\rcov{\lat{\ell_0}}^2$. For the integer chain, if this residue
is additionally uniform on the balanced set $\lat1\cap\vor{\ell_0}$, then
\begin{align}\label{eq:resenergy}
T_{\ell_0}
=n_{\mathrm b}\sigma^2(\lat{\ell_0})
\bigl(1-\rho^{-2(\ell_0-1)}\bigr),
\end{align}
with $T_1=0$.

For client $k$, define the one-block reconstruction error before DP noise as
\begin{align}
\mathbf E_k^{(\ell_0)}
&=Q_{\lat{\ell_0}}(\mathbf x_k)-\mathbf d_k-\mathbf v_k\\
&=\mathbf e_k-\mathbf r_k^{(\ell_0)}.
\end{align}
By \eqref{eq:telescope},
$\mathbf r_k^{(\ell_0)}=\sum_{\ell<\ell_0}\mathbf x_k^\ell$, so the residue
is exactly the withheld refinement.

\begin{theorem}[End-to-end MSE]\label{thm:mse}
Let $\widetilde{\mathbf s}_{\Delta,\ell_0}$ be the ideal estimator obtained
from \eqref{eq:recon} by replacing every authorized modular residue by its
true noisy lattice sum. By Lemma~\ref{lem:carry}, this ideal estimator equals
the actual estimator on the zero-carry event. Under the concentration rule of
Theorem~\ref{thm:margin}(b), that event has probability at least
$1-\delta_w-\delta_z$ when the rule is applied to all released layers; the
moments below are moments of the ideal estimator, not unconditional moments of
the modular estimator outside that event. Suppose client data-and-dither pairs
are independent across clients and, for $\ell_0>1$, the skipped refinement of
each client has zero mean, so $\E[\mathbf E_k^{(\ell_0)}]=\mathbf0$. At full
resolution the residue is zero and this condition follows from subtractive
dithering alone. Together with the independent zero-mean DP noises specified in
the construction, these assumptions make the cross terms vanish. For one
block,
\begin{align}\label{eq:msedec}
\E\|\widetilde{\mathbf s}_{\Delta,\ell_0}-\mathbf s_\Delta\|^2
=\frac1{a^2}\Bigg[&
\sum_{k=1}^K\E\|\mathbf E_k^{(\ell_0)}\|^2\notag\\[-0.4ex]
&+K\sum_{r:\,\mathcal G_r\ \mathrm{authorized}}
 n_{\mathrm b}\bar s_r^2\Bigg].
\end{align}
The per-client error term behaves as follows under truncation:
\begin{enumerate}[leftmargin=1.4em,label=(\alph*)]
\item \emph{Full resolution.}
For $\ell_0=1$,
$\E\|\mathbf E_k^{(1)}\|^2=n_{\mathrm b}\sigma^2(\lat1)$
exactly, without a distributional assumption on the update.
\item \emph{Any truncation.} For identically distributed client blocks, with $T_{\ell_0}$ defined above,
\begin{align}\label{eq:msebounds}
&\left(\sqrt{T_{\ell_0}}-
\sqrt{n_{\mathrm b}\sigma^2(\lat1)}\right)^2
\le\E\|\mathbf E_k^{(\ell_0)}\|^2\nonumber\\
&\le
\left(\sqrt{T_{\ell_0}}+
\sqrt{n_{\mathrm b}\sigma^2(\lat1)}\right)^2.
\end{align}
The bounds are two-sided because $\mathbf e_k$ and
$\mathbf r_k^{(\ell_0)}$ are functions of the same data and dither pair,
hence correlated in general; only their individual second moments are
known, so their inner product is bounded rather than evaluated. They are
distribution-free once the actual $T_{\ell_0}$ is used, and apply to
non-identically distributed clients with
$T_{k,\ell_0}=\E\|\mathbf r_k^{(\ell_0)}\|^2$ per summand. Always
$T_{\ell_0}\le\rcov{\lat{\ell_0}}^2$.
\item \emph{Scaling under uniform residues.} For the integer chain
under the additional uniform-residue condition, which leads to
\eqref{eq:resenergy}, the two sides of \eqref{eq:msebounds} are
$\Theta(\rho^{2(\ell_0-1)})$ for $\ell_0\ge2$ (with fixed odd
$\rho$). Without that condition, only the universal upper scaling
$O(\rho^{2(\ell_0-1)})$ follows from the support bound; a matching
lower scaling need not hold.
\end{enumerate}
\end{theorem}

\begin{proof}
Appendix~\ref{app:mse}.
\end{proof}

Thus $\lat1$ fixes the full-resolution quantization floor, while coarser
releases add a truncation term whose geometric scaling requires the stated
uniform-residue condition. DP noise contributes separately through every
authorized stage.
\vspace{-5pt}
\section{Numerical Evaluation}
\label{sec:sims}

The evaluation separates claims that hold for the code alone from claims that
depend on the data. Exact finite-alphabet computation settles the carry margins
and the digit-compatibility condition; controlled sources settle the distortion
bounds and the privacy accounting; federated experiments on MNIST and CIFAR-10
settle the digit model and the one-upload trade-off against matched baselines.
\vspace{-10pt}
\subsection{Carries, Margins, and Digit Compatibility}
\label{sub:exact}

\begin{figure*}[t]
\centering
\includegraphics[width=.7\textwidth]{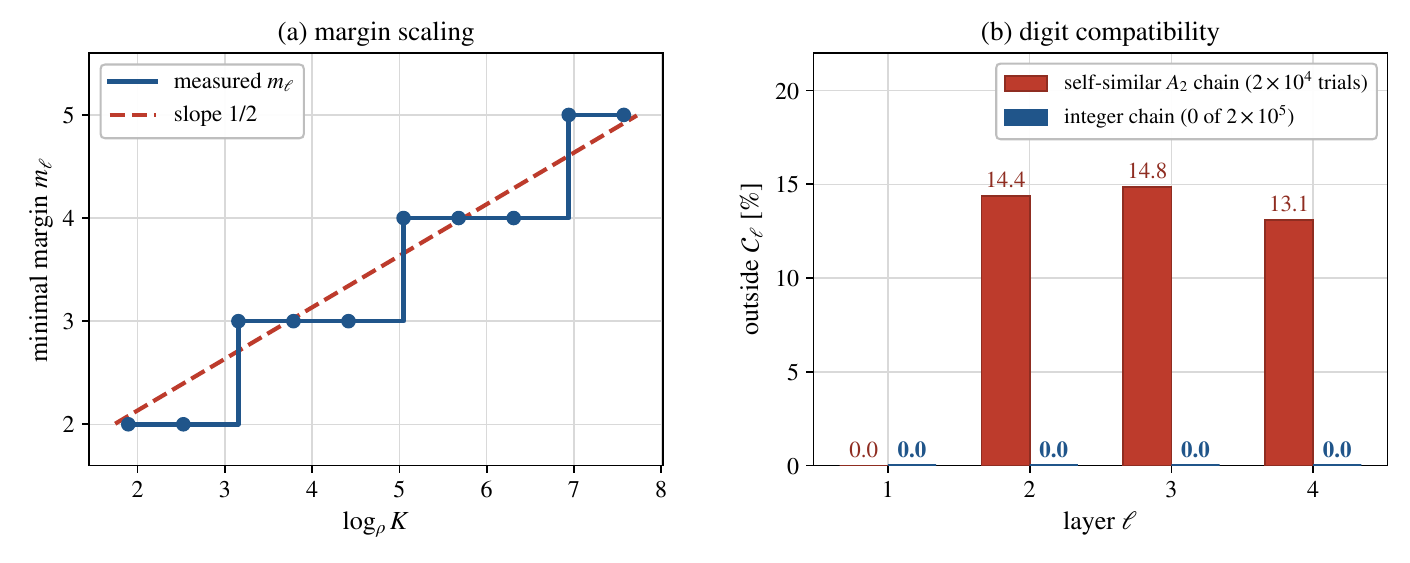}
\vspace{-10pt}
\caption{Exact finite-alphabet validation, $\rho=3$, $L=4$. (a) Smallest
feasible per-layer margin at carry probability at most $10^{-3}$ against
$\log_\rho K$; the staircase reflects integrality of $m_\ell$, and a linear
fit to it has slope $0.548$ against the predicted $1/2$. (b) Rate at
which $Q_{\lat\ell}-Q_{\lat{\ell+1}}$ leaves the centered alphabet
$\mathcal C_\ell$, over $2\times10^5$ random representable points for the
integer chain and $2\times10^4$ for the $A_2$ chain: identically zero for the
odd-radix integer chain (annotated zeros), nonzero from layer~$2$ for a merely
self-similar $A_2$ chain.}
\label{fig:exact}
\end{figure*}

The scalar chain uses $\rho=3$, $L=4$ and client digits i.i.d.\ uniform on
$\{-1,0,1\}$, matching the prior of Theorem~\ref{thm:converse}; digit-sum
distributions are computed by exact polynomial convolution, so the carry
probabilities carry no Monte Carlo error. Fig.~\ref{fig:exact}(a) shows the smallest feasible margin. It grows from $2$ at
$K=8$ to $5$ at $K=4096$ in integer steps, and regressing it on $\log_3K$ gives
slope $0.548$ against the predicted $1/2$. Theorem~\ref{thm:margin}(b) asserts
$\tfrac12\log_\rho K+O(1)$, so it is the slope and not the value that the
measurement tests: the margins lie about $1.1$ steps above
$\tfrac12\log_\rho K$, and the reference line is drawn with slope $1/2$ through
the data for that reason. Realized carry probabilities never exceed
$9.75\times10^{-6}$. Because the rate premium of Proposition~\ref{prop:rate} is
$(L-1)m_\ell\log_2\rho$, it is this staircase up to a constant and is not
plotted separately; numerically it is $19.02$ bit per dimension at $K=1000$
against the asymptotic $\tfrac{L-1}{2}\log_2K=15.00$, the offset being
integrality rather than a failure of the prediction.

Fig.~\ref{fig:exact}(b) tests Assumption~\ref{ass:compatible} directly. Over
$2\times10^5$ random representable points the integer chain satisfies
telescoping, alphabet membership and the quantizer-difference form of
\eqref{eq:telescope} without a violation. On a self-similar hexagonal $A_2$
chain the telescoping identity still holds, being algebraic, but membership
fails from the second layer onward, at $14.4\%$, $14.8\%$ and $13.1\%$ of
$2\times10^4$ trials. Self-similarity therefore does not deliver digit compatibility.
\vspace{-10pt}
\subsection{When Layering Is Free, and When It Is Not}
\label{sub:priv}

\begin{figure*}[t]
\centering
\includegraphics[width=.7\textwidth]{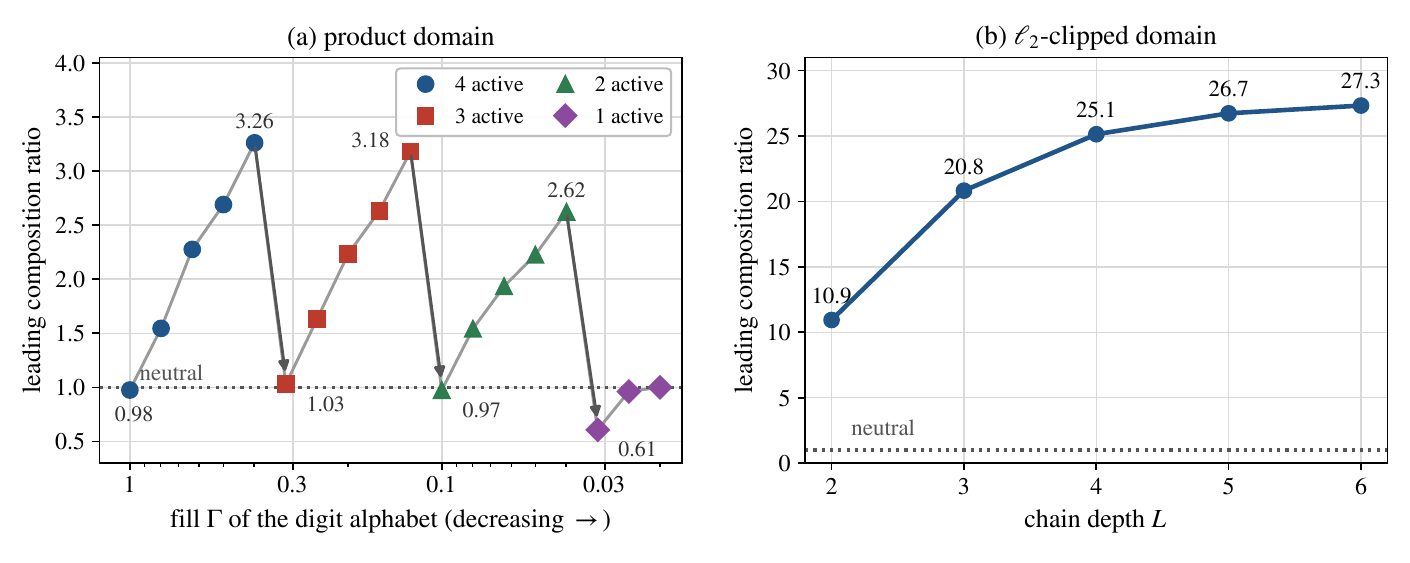}
\vspace{-10pt}
\caption{Leading composition ratio
$\varepsilon^{\rm lead}_{\rm lay}/\varepsilon^{\rm lead}_{\rm flat}$ at equal
noise budget. (a) Product domain as the fill $\Gamma$ of the digit alphabet is
reduced; markers give the number of layers still active, and arrows mark the
transitions at which a layer empties. (b) Globally $\ell_2$-clipped domain
against chain depth, where no layer ever becomes inactive.}
\label{fig:privacy}
\end{figure*}

Proposition~\ref{prop:cost} predicts neutrality only when the admissible domain
fills the digit alphabet in every coordinate. Fig.~\ref{fig:privacy}(a) measures
the leading ratio as that condition is relaxed, and the behaviour is not
monotone. The ratio is $0.975$ at full fill, confirming the predicted
neutrality, then rises to $3.26$ as the fill drops to $0.40$ with all four
layers active. At that point the coarsest layer becomes inactive, its
sensitivity vanishes, it leaves the privacy sum, and the ratio falls back to
$1.03$. The pattern repeats twice more, at $3.18\to0.97$ and $2.62\to0.61$,
until a single layer remains and the ratio is trivially $1$ because one stage
is the flat mechanism.

The penalty of layering is therefore governed not by how restricted the domain
is but by how poorly the \emph{active} layers are filled. A chain whose depth is
matched to the data sits near the bottom of a tooth rather than the top, which
is a design lever the analysis does not expose and which the occupancy
measurements below make actionable.

The globally $\ell_2$-clipped domain assumed by \cite{ddgauss,skellam} is a
different regime, Fig.~\ref{fig:privacy}(b). A ball never fills the coarse
alphabets, no layer becomes inactive, and the ratio grows with depth to $25.1$
at $L=4$, saturating near $27$. This factor, not the rate premium, is the
dominant practical cost of layering under the clipping the standard baselines
use. The caveat concerns the admissible domain, not the adjacency relation,
which is client-level replacement throughout.
\vspace{-10pt}
\subsection{Distortion and the Two Error Sources}
\label{sub:mse}

\begin{figure*}[t]
\centering
\includegraphics[width=.7\textwidth]{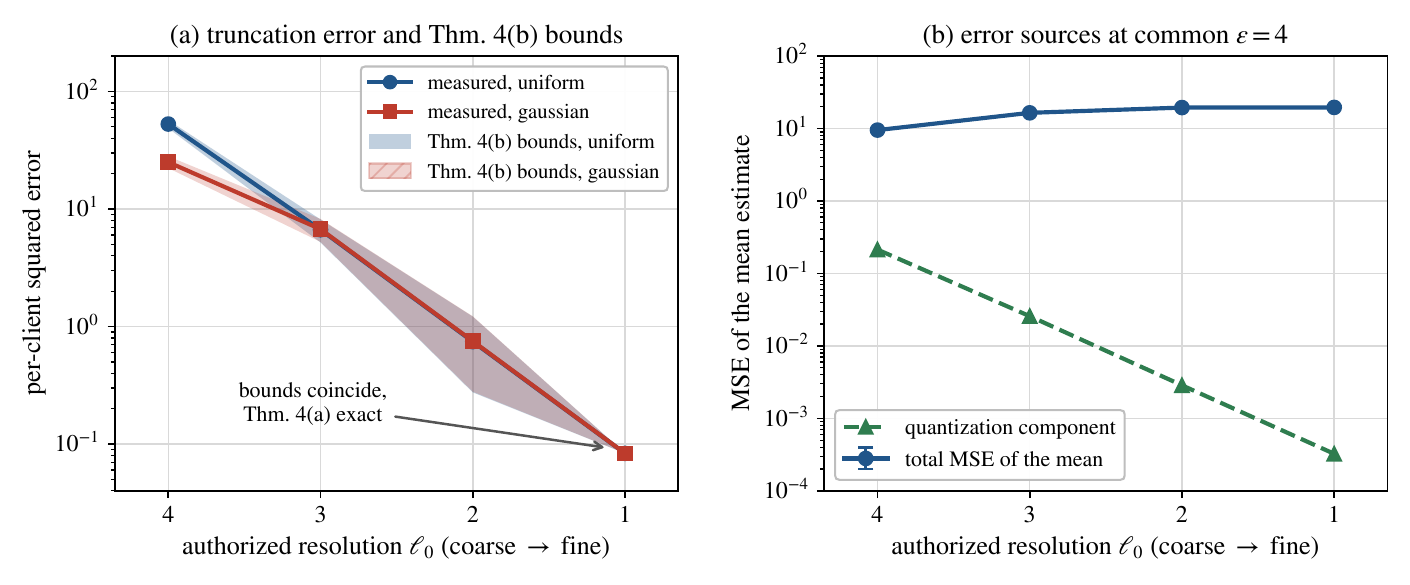}
\vspace{-10pt}
\caption{(a) Per-client truncation error against the two-sided bounds of
Theorem~\ref{thm:mse}(b), drawn as shaded bands for $\ell_0\ge2$ where they
have width; at $\ell_0=1$ the bounds coincide and Theorem~\ref{thm:mse}(a) is
exact. (b) Mean-square error of the mean estimate with every depth recalibrated
to $\varepsilon=4$, separated into its total and its quantization component.
Bars are $95\%$ confidence intervals; where not visible they are smaller than
the markers.}
\label{fig:mse}
\end{figure*}

Fig.~\ref{fig:mse}(a) places the measured per-client error against
Theorem~\ref{thm:mse}(b). Every measurement lies inside the bounds and the error
falls by close to $\rho^2=9$ per refinement level, from $5.28\times10^{1}$
(uniform) and $2.50\times10^{1}$ (Gaussian) at $\ell_0=4$ to
$8.3\times10^{-2}$ at $\ell_0=1$. That the bounds hold for a
Gaussian source, which violates the uniform-residue condition, is the point:
they are distribution-free. The closed form for $T_{\ell_0}$ needs that
condition and deviates from the measured second moment by $13.0\%$ (uniform)
and $58.6\%$ (Gaussian) at the coarsest level, so accuracy statements should
use the measured $T$ or the bounds.

Fig.~\ref{fig:mse}(b) recalibrates every depth to a common $\varepsilon=4$ and
separates the two contributions. The quantization component falls monotonically
by nearly three orders of magnitude, from $0.214$ to $3\times10^{-4}$, exactly
as (a) predicts. The total error moves the other way, rising from $9.52$ at
$\ell_0=4$ to $19.48$ at $\ell_0=2$ and then saturating; the final step to
$19.58$ lies within the $\pm0.61$ confidence interval, so the last refinement
neither helps nor measurably hurts. Each additional authorized stage contributes
an independent noise draw whose variance exceeds the truncation error it
removes, until the truncation term is small enough that neither dominates.
Refinement improves fidelity only once the privacy budget is loosened. This is a
property of the independent-stage-noise construction, not a lower bound on
progressive release. No carry occurred at the provisioned margins in any run.

\subsection{Federated Evaluation}\label{sec:fedeval}
\label{sub:setup}

\emph{Setup.} MNIST and CIFAR-10 are used with their standard splits, processed
by the unsupervised convolutional feature pipeline of Coates \emph{et
al.}~\cite{coates}: whitened random patches, a $k$-means dictionary, triangle
activation, spatial pooling and PCA. MNIST uses $256$ centroids, $2\times2$
pooling and PCA-$128$; CIFAR-10 uses $512$ centroids, $3\times3$ pooling and
PCA-$384$. The model is multinomial logistic regression, so $d=1290$ and
$d=3850$, and the chain is $\rho=3$, $L=4$. The scale is deliberate: the
claims under test concern the aggregation mechanism rather than the learner,
and a convex model at full participation isolates the privacy accounting from
optimization effects while keeping every run reproducible on CPU. Training is FedSGD with server
momentum $0.9$ \cite{mcmahan}, $K=1000$ clients at full participation, non-IID
Dirichlet($0.5$) partitioning, $150$ rounds and per-client clipping at $C=1$.
Reported $\varepsilon$ values are run totals: the per-round mechanism is
composed $T=150$ times in the R\'enyi accountant before conversion, with no
subsampling amplification since participation is full. Results are means over
three seeds.

\emph{Baselines.} The non-private arm is plain FedSGD \cite{mcmahan}; the
compression-only arm is single-layer subtractive-dithered lattice quantization
after UVeQFed \cite{uveqfed} with no DP noise, isolating the cost of the
representation. The private arms are the distributed discrete Gaussian
\cite{ddgauss} and Skellam \cite{skellam} mechanisms under SecAgg
\cite{bonawitz}, Skellam credited with its leading R\'enyi term only, which
favors that baseline. All private arms share the lattice granularity $h$, so
resolution is matched and every accuracy difference is attributable to the
privacy accounting rather than to the lattice code.

\emph{Matched functionality.} The only baseline serving $R$ independently
authorized resolutions is repeated flat SecAgg, run once per resolution at the
granularity that resolution requires. An execution serving $\ell_0$ transmits
$(L-\ell_0+1+m)\log_2\rho$ bit per dimension rather than $(L+m)\log_2\rho$, and
the $R$ mechanisms compose, so the baseline delivers the same $R$ numbers only
by re-executing and recontacting the clients. We sweep $R\in\{1,2,4\}$ with
stages $\{4,3,2,1\}$, $\{4,3\},\{2,1\}$, and one layer per stage, at equal
$(\varepsilon,\delta)$ with $\delta=10^{-6}$.

\emph{The digit model.} Assumption~\ref{ass:qu} is the one modeling assumption
in the achievability argument, and it concerns the data. On real clipped
gradients the largest normalized digit bias
$|\mathbb{E}[a^\ell]|/\mathrm{sd}(a^\ell)$ is $0.031$ on MNIST and $0.023$ on
CIFAR-10 across all four layers, so its zero-mean part is accurate for
gradient-type updates around a shared reference. It does not test independence,
since clients in a round share the current model; the worst-case rule of
Theorem~\ref{thm:margin}(a) needs neither part, at a margin growing as $K$
rather than $\sqrt K$. Layer occupancy, the fraction of coordinates with a
nonzero digit, is $(0.165,0.015,1.6\times10^{-3},1.9\times10^{-4})$ on MNIST and
$(0.124,0.003,3.9\times10^{-4},6.8\times10^{-5})$ on CIFAR-10. On CIFAR-10 the
finest provisioned layer is essentially unoccupied, so three of the four layers
carry the information; on MNIST all four are occupied, the finest only barely.
By Fig.~\ref{fig:privacy}(a) an unoccupied layer is not merely wasted rate but
also moves the chain up a tooth of the privacy cost: depth should be provisioned
from measured occupancy, not from the nominal dynamic range.

\begin{table}[t]
\centering
\caption{Matched functionality at $\delta=10^{-6}$: one PSA upload serving $R$
authorized resolutions versus $R$ repeated flat executions. Accuracy is the mean
over three seeds; payload is bit per dimension at $K=1000$.}
\label{tab:matched}
\small
\setlength{\tabcolsep}{3.5pt}
\begin{tabular}{llcccc}
\hline
 & & \multicolumn{2}{c}{MNIST} & \multicolumn{2}{c}{CIFAR-10} \\
$\varepsilon$ & $R$ & PSA & flat $\times R$ & PSA & flat $\times R$ \\
\hline
$6$  & 1 & 0.771 & 0.771 & 0.274 & 0.274 \\
     & 2 & 0.650 & 0.718 & 0.195 & 0.227 \\
     & 4 & 0.492 & 0.626 & 0.154 & 0.206 \\
\hline
$30$ & 1 & 0.910 & 0.910 & 0.429 & 0.429 \\
     & 2 & 0.854 & 0.883 & 0.354 & 0.389 \\
     & 4 & 0.709 & 0.839 & 0.248 & 0.365 \\
\hline
\multicolumn{2}{l}{payload $R=1$} & 12.68 & 12.68 & 12.68 & 12.68 \\
\multicolumn{2}{l}{payload $R=2$} & 19.02 & 22.19 & 19.02 & 22.19 \\
\multicolumn{2}{l}{payload $R=4$} & 31.70 & 41.21 & 31.70 & 41.21 \\
\hline
\end{tabular}
\end{table}

\begin{figure}[t]
\centering
\includegraphics[width=0.88\columnwidth]{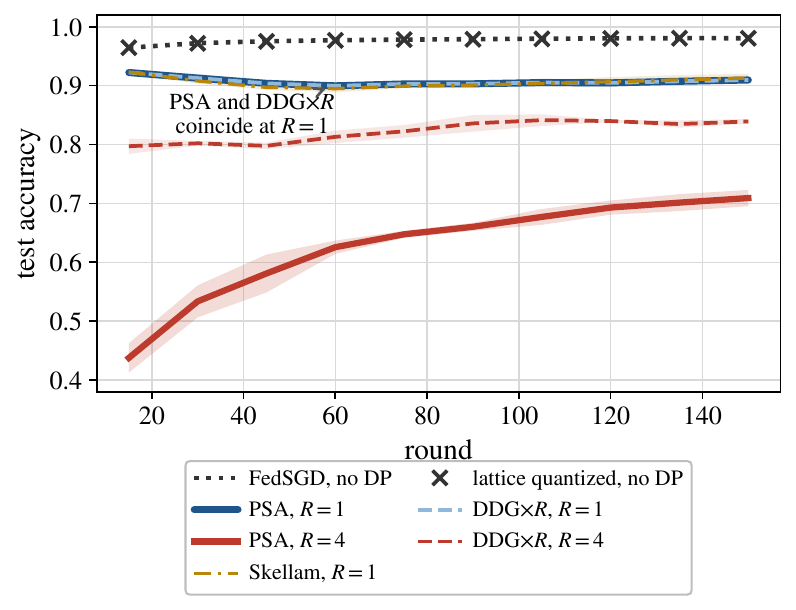}
\vspace{-10pt}
\caption{MNIST federated training at matched $\varepsilon=30$,
$\delta=10^{-6}$, three seeds (shaded: one standard deviation, below line
width except at $R=4$). Solid curves are PSA, dashed the distributed
discrete Gaussian re-executed once per resolution (DDG$\times R$), dash-dotted
the Skellam mechanism, which serves a single release; the
extremes $R=1$ and $R=4$ are shown, with $R=2$ in
Table~\ref{tab:matched}. Two coincidences are results rather than plotting
artifacts: the lattice-quantized arm tracks the non-private reference, and at
$R=1$ the dashed curve lies on the solid one at every logged round.}
\label{fig:fl}
\end{figure}

\emph{Results.} Table~\ref{tab:matched} reports the comparison. At $R=1$ PSA and
the flat mechanism agree to three decimals at every privacy level on both
datasets, and the payloads coincide at $12.68$ bit per dimension: with one
authorization stage PSA provisions one modulus, exactly as the flat mechanism
does. This is the falsification test of Section~\ref{sub:setup} and it passes;
any advantage here would have indicated an accounting error rather than a
property of the scheme. Fig.~\ref{fig:fl} shows that the $R=1$ curves agree at
every logged round, not merely at convergence, so the neutrality is structural.
The Skellam arm, which serves a single release, tracks the $R=1$ pair at matched
budgets, reaching $0.779$, $0.847$ and $0.913$ on MNIST and $0.229$, $0.352$ and
$0.430$ on CIFAR-10 at $\varepsilon=6,12,30$; at $\varepsilon=6$ its
unprovisioned modulus wraps ($4.7\times10^{3}$ carry events on MNIST,
$5.3\times10^{4}$ on CIFAR-10), which costs it four accuracy points on
CIFAR-10, whereas PSA's provisioned margins admit no carry in any run.

For $R>1$ the two directions separate. PSA's finest consumer reconstructs the
sum of all authorized stages and accumulates the noise of every one, whereas a
consumer of the repeated flat pipeline reads a single execution and sees one
noise draw; stage sensitivities do not shrink proportionally, because a small
change of the input can alter every layer symbol
(Remark~\ref{rem:notele}). At $\varepsilon=30$ PSA reaches $0.854$ against
$0.883$ at $R=2$ and $0.709$ against $0.839$ at $R=4$. Communication runs the
other way: PSA uploads once with one modulus per stage, giving $12.68$, $19.02$
and $31.70$ bit per dimension at $R=1,2,4$, against $12.68$, $22.19$ and $41.21$
for the repeated pipeline, so PSA is cheaper for every $R>1$, at $0.86\times$
the payload for $R=2$ and $0.77\times$ for $R=4$. Privacy composition runs the
same way: serving four resolutions costs $\varepsilon=2.00$ for one PSA
execution against $\varepsilon=4.29$ for four flat executions, a factor $2.14$,
with the flat totals $2.00$, $2.92$, $3.65$ and $4.29$ across $R=1,2,3,4$ while
the PSA total stays at $2.00$. The asymmetry is structural: PSA is charged for
the layers it reveals, and a coarser resolution is a prefix of layers already
revealed, so serving it costs nothing further, whereas each flat resolution is a
separate execution composed in its own right.

The representation itself is nearly free. The compression-only arm attains
$0.9804\pm0.0003$ against $0.9806\pm0.0003$ for the uncompressed non-private
reference on MNIST, and $0.6073\pm0.0005$ against $0.6086\pm0.0010$ on
CIFAR-10, so every accuracy difference among the private arms is attributable
to the privacy accounting rather than to the lattice code.

The absolute accuracies identify the operating region in which progressive
release is inexpensive. The configuration is deliberately conservative: full
participation over $150$ rounds admits no amplification by subsampling, so every
arm, including the flat baseline at $0.771$ for $\varepsilon=6$, sits well below
the non-private reference of $0.981$; amplification or fewer rounds would raise
all arms together. Within this setting $R=1$ is free by construction and $R=2$
costs between $3$ and $7$ accuracy points on MNIST, narrowing as the budget
loosens. The cost of the finest release grows with the number of independent
authorizations, so serving four is best matched to looser budgets. Taken
together with the payload and composition figures, the case for PSA rests on the
enforcement guarantee it provides, not on a drop-in accuracy or bandwidth
improvement.
\vspace{-10pt}
\subsection{The Rotated Domain: Closing the Layering Penalty}
\label{sec:rotdomain}

\begin{figure}
\centering
\includegraphics[width=0.7\columnwidth]{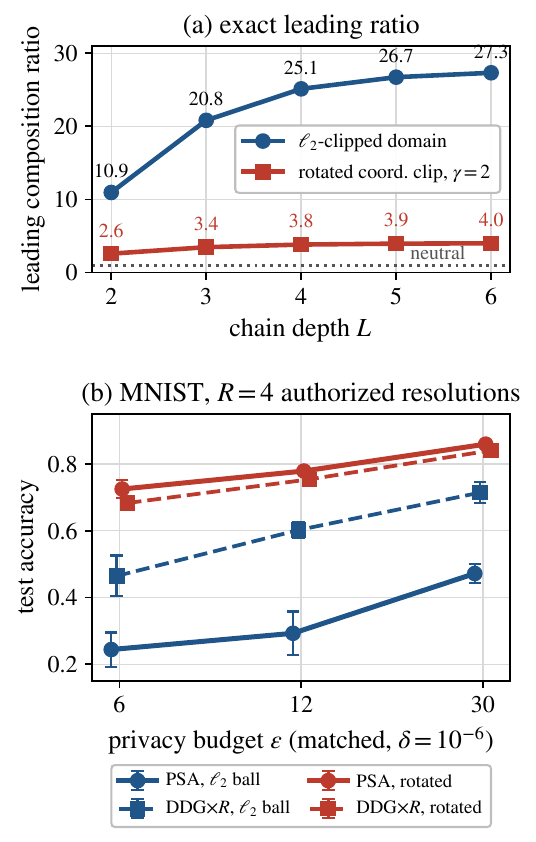}
\vspace{-10pt}
\caption{The operating regime as a design choice, $\gamma=2$. (a) Exact
leading composition ratio of the $\ell_2$-clipped domain against the rotated
coordinatewise-clipped domain, with the surviving $\ell_2$ bound retained in
the sensitivity accounting. (b) Federated accuracy at $R=4$ authorized
resolutions and matched $(\varepsilon,\delta)$, three seeds (bars: one
standard deviation), with the same rotation granted to the repeated flat
baseline; the unrotated PSA arm uses this experiment's conservative stage
bounds and is a within-experiment reference (see text).}
\label{fig:rot}
\end{figure}

The penalty of Fig.~\ref{fig:privacy}(b) is a property of the operating
regime, and the regime is a protocol choice. Instantiating the public common
rotation of Section~\ref{sec:scheme} in randomized-Hadamard form, the
standard flattening transform also used by~\cite{ddgauss}, and then
clipping each rotated coordinate at $\gamma C/\sqrt{d}$ makes the admissible
set a product box that nearly fills the digit alphabet. Retaining the
surviving $\ell_2$ bound in the sensitivity accounting, the leading
composition ratio at $L=4$ falls from $25.1$ to $3.8$ at $\gamma=2$
(Fig.~\ref{fig:rot}(a)), while the measured coordinate-clipping rate stays
below $10^{-4}$. Fig.~\ref{fig:rot}(b) reports matched-$(\varepsilon,\delta)$
training under this instantiation, with the same transform granted to the
baseline. At $\varepsilon=30$ and $R=4$, rotated PSA attains $0.859$ on MNIST
against $0.714$ for the repeated flat pipeline, and $0.380$ on CIFAR-10
against $0.245$; over the full $(\varepsilon,\text{dataset})$ sweep this
margin exceeds $0.08$ in every cell and every seed. Against the rotated
variant of the baseline ($0.841$ and $0.370$) the paired per-seed differences
are positive on MNIST but small and of mixed sign on CIFAR-10, so we claim
statistical parity rather than superiority: in this regime one PSA upload
serves four resolutions at no accuracy cost relative to matched
functionality. The unrotated PSA arm ($0.472$ at $\varepsilon=30$ on MNIST)
is a within-experiment reference only: this experiment's accountant bounds
each ball-regime stage conservatively, with an effective layered-to-flat
noise-variance ratio near $180$ against the per-layer value $25.1$ of
Fig.~\ref{fig:privacy}(b), so its level understates ball-regime PSA and is
not comparable to the $0.709$ of Table~\ref{tab:matched}, whose accountant
uses the tighter per-layer sensitivities. The $R=1$ consistency check of
Section~\ref{sec:fedeval} holds bitwise in the rotated regime as well, no
carry occurred in any of the $156$ runs, and the payload and composition
advantages of Table~\ref{tab:matched} are unchanged. These runs use an
independent re-implementation of the feature pipeline, with non-private
references $0.949$ (MNIST) and $0.607$ (CIFAR-10), so all comparisons are
within the experiment.
\vspace{-5pt}
\section{Discussion and Limitations}
\label{sec:disc}

PSA separates three guarantees. SecAgg hides individual client uploads; DP
limits what an authorized aggregate reveals; and the release gate controls
when a stage aggregate becomes available. If the gate is compromised, delayed
authorization is lost and the construction reduces to ordinary DP-secure
aggregation of the stage sums. If DP is omitted, staged disclosure remains but
each authorized aggregate receives no DP guarantee. Pairwise SecAgg masks cannot replace the release gate because
they cancel in the aggregate.

The arithmetic-payload cost is tied to authorization granularity. A one-upload
protocol must carry every stage that may later be authorized, and PSA provisions
one overflow margin per such modular sum. Under the i.i.d.\ uniform-digit prior,
Corollary~\ref{cor:price} shows that the
$\tfrac12\log_\rho K$ concentration-order margin is unavoidable within the
layerwise modular interface. This is not a lower bound on jointly encoded,
interactive, or richer interfaces. A deployment can trade release granularity
against arithmetic payload through the number of stages $R$.

The guarantees also have distinct scope restrictions. Exact layered
reconstruction requires Assumption~\ref{ass:compatible}; the canonical integer
chain is a proved instance, whereas self-similarity alone is insufficient. The
smaller concentration margin additionally requires Assumption~\ref{ass:qu};
without it the worst-case rule of Theorem~\ref{thm:margin}(a) applies. The
product-integer discrete-Gaussian RDP bound does not automatically extend to
nonproduct lattices: such an extension requires the corresponding convolution
and smoothing analysis described in Remark~\ref{rem:general}. Finally,
Proposition~\ref{prop:cost} is a leading accounting comparison whose neutral
case depends on the admissible domain.

Production realization of the abstract release gate requires mechanisms not
constructed here, including binding pad recovery to the realized active set,
audited key derivation and domain separation, rollback protection, and
malicious-input validation. These are interface requirements rather than
consequences of the lattice analysis. The multi-stage accuracy loss observed
with independent stage noise under $\ell_2$ clipping is likewise
construction-specific rather than a lower bound: the rotated-domain
instantiation of Section~\ref{sec:rotdomain} removes it empirically at
$R=4$, and correlating noise across progressive releases remains a further
candidate.
\vspace{-4pt}
\section{Conclusion}
\label{sec:concl}

We introduced and formalized progressive-resolution secure aggregation as a one-upload functionality that enables successively finer resolutions of the same federated aggregate to be authorized after the clients have uploaded, without renewed client participation. Under independent zero-mean layer symbols and normalized
aggregate-noise scale $O(\sqrt K)$, its sufficient overflow margin is
$\tfrac12\log_\rho K+O(1)$ radix steps. Under the separate i.i.d.\ uniform-digit
prior, the converse establishes the same necessary order within the layerwise
modular interface; it is not a universal lower bound on other protocol
interfaces. Separate release pads hide unauthorized stage sums, while the
arithmetic-payload cost grows with the number of separately sealed moduli.

The numerical results support these qualified claims. Exact finite-alphabet
calculations exhibit the predicted margin and payload scaling; the $R=1$
federated consistency check makes PSA coincide with the matched flat mechanism;
and the compression-only lattice representation changes accuracy by at most
$0.13$ percentage points in the reported runs. With several authorized
resolutions under $\ell_2$ clipping, independent stage noise lowers
finest-resolution accuracy while the one-upload construction reduces arithmetic
payload relative to repeated flat executions. In the rotated
coordinatewise-clipped regime of Section~\ref{sec:rotdomain}, PSA reaches
statistical parity with the matched rotated baseline while retaining the
one-upload progressive-release functionality.

The exact layered construction is proved for the balanced integer chain and
does not follow from self-similarity alone, as illustrated by the observed
$A_2$ compatibility failures. Extending the discrete-Gaussian privacy constants
to nonproduct lattices, handling malicious inputs, and realizing production key
management for the release gate remain open. The multi-stage accuracy gap
associated with independent stage noise under $\ell_2$ clipping also remains a
limitation of the present construction and motivates the investigation of
alternative noise couplings across progressive releases.
\appendices
\vspace{-5pt}
\section{Proof of Theorem~\ref{thm:margin}}
\label{app:margin}

Fix $\ell$. A point of $\lat\ell$ with norm below
$\rpack{\lat{}^{(\ell)}}=\rho^{m_\ell}p_{\ell+1}$ lies in the Voronoi cell
of the modulus lattice, so by Lemma~\ref{lem:carry} it suffices to bound
$\|\mathbf S_\ell+\mathbf Z_\ell\|$. For (a), every
$\mathbf x_k^\ell\in\mathcal C_\ell$ has norm at most $r_{\ell+1}$, so
$\|\mathbf S_\ell\|\le K r_{\ell+1}$; on $\|\mathbf Z_\ell\|\le u_\ell$
condition \eqref{eq:worstmargin} forces zero carry, the complement having
probability at most $\delta_z/L$. For (b), Assumption~\ref{ass:qu} gives
independent zero-mean summands in $\R^{n_{\mathrm b}}$ bounded by
$r_{\ell+1}$ with total second moment $V_\ell$, and the vector Bernstein
bound of Pinelis \cite{pinelis} gives
\begin{align}
\Prob[\|\mathbf S_\ell\|\ge t]
\le2\exp\!\left(
-\frac{t^2}{2V_\ell+\frac23r_{\ell+1}t}
\right).
\end{align}
The bound has no closed-form inverse, but a sufficient $t$ covers its
variance- and range-dominated regimes by its two terms:
\begin{align}
t=\sqrt{2V_\ell\log\frac{4L}{\delta_w}}
+\frac23r_{\ell+1}\log\frac{4L}{\delta_w},
\end{align}
where $4L$ rather than $2L$ absorbs the leading factor $2$, making
$\Prob[\|\mathbf S_\ell\|\ge t]\le\delta_w/L$; with the noise-tail
condition and the triangle inequality this yields \eqref{eq:conc}. A union
bound gives $\delta_w+\delta_z$ when part (b) is used on all layers; part (a)
uses only the $\delta_z$ budget. The choice \eqref{eq:uint} serves both parts,
following from the sub-Gaussian proxy by a two-sided scalar tail bound at level
$\delta_z/(n_{\mathrm b}L)$ per coordinate and
$\|\mathbf z\|_2\le\sqrt{n_{\mathrm b}}\|\mathbf z\|_\infty$.
\hfill$\blacksquare$
\vspace{-5pt}
\section{Proof of Theorem~\ref{thm:converse}}
\label{app:converse}

Fix $\ell$ and give the decoder, as genie information, the exact sums
$\{S_j:j\ne\ell\}$, which can only increase its success probability. Since
$X=\sum_j\rho^{j-1}S_j$ and every other term is now known, recovering $X$
is equivalent to recovering $S_\ell$. Digits are independent across layers,
so the genie leaves the law of $S_\ell$ unchanged and its only
data-dependent observation remains $t_\ell=S_\ell\bmod M_\ell$; bounding
that single-layer problem bounds the original.

Let $p_K(s)=\Prob[S_\ell=s]$. For any decoder based on the residue, the
optimal success probability is
\begin{align}
\sum_{r\in\Z/M_\ell\Z}
\max_{s:\,s\equiv r\ (\mathrm{mod}\ M_\ell)} p_K(s)
\le M_\ell\sup_s p_K(s).
\end{align}
The balanced digit law is bounded, has variance $\sigma_\rho^2$, and is
aperiodic, its support lying in no proper sublattice of $\Z$, as required
for the lattice local limit theorem at spacing one \cite{gnedenko}. Hence
there are constants $C,K_0$ with $\sup_s p_K(s)\le C/(\sigma_\rho\sqrt K)$ for
$K\ge K_0$. Capping at one proves \eqref{eq:convbound}, and a success
probability of at least $1-\delta$ forces the bound itself to be at least
$1-\delta$; substituting $M_\ell=\rho^{1+m_\ell}$ gives
\eqref{eq:convmargin}. \hfill$\blacksquare$
\vspace{-5pt}
\section{Proof of Theorem~\ref{thm:mse}}
\label{app:mse}

Write \eqref{eq:recon} over the authorized stages, as licensed by
Remark~\ref{rem:schedule}. Every released stage is unwrapped to its true
integer sum, so subtracting \eqref{eq:target} from \eqref{eq:recon} gives
\begin{align}
a(\widetilde{\mathbf s}_{\Delta,\ell_0}-\mathbf s_\Delta)
=\sum_k\mathbf E_k^{(\ell_0)}
+\sum_k\sum_{r:\,\mathcal G_r\ \mathrm{authorized}}\mathbf z_k^{r}.
\end{align}
Cross terms vanish as noted with the theorem, giving \eqref{eq:msedec};
part (a) is the case $\ell_0=1$, where $\mathbf r_k^{(1)}=0$.

For general $\ell_0$,
\begin{align}
\E\|\mathbf e_k-\mathbf r_k^{(\ell_0)}\|^2
=&\ n_{\mathrm b}\sigma^2(\lat1)+T_{\ell_0}\\
&-2\E\langle\mathbf e_k,\mathbf r_k^{(\ell_0)}\rangle.
\end{align}
Cauchy--Schwarz bounds the last term by
$2\sqrt{n_{\mathrm b}\sigma^2(\lat1)T_{\ell_0}}$, proving
\eqref{eq:msebounds}, and the support bound follows from
$\mathbf r_k^{(\ell_0)}\in\vor{\ell_0}$. Under the uniform-residue law
\eqref{eq:resenergy} follows coordinatewise from the variance of the
balanced set of $\rho^{\ell_0-1}$ integers, and since
$\sigma^2(\lat{\ell_0})=\rho^{2(\ell_0-1)}\sigma^2(\lat1)$ both bounds have
geometric order for $\ell_0\ge2$; without it the support bound gives
$O(\rho^{2(\ell_0-1)})$ but no matching lower bound. \hfill$\blacksquare$

\end{document}